 \documentclass[final,3p,times]{elsarticle}

\usepackage{lineno}
\usepackage{amssymb}
\usepackage{amsmath}
\usepackage{amsthm}

\newtheorem{theorem}{Theorem}[section]

\newtheorem{lemma}[theorem]{Lemma}
\newtheorem{definition}[theorem]{Definition}

\newcommand{\RR}{\mathbb{R}}
\newcommand{\SSS}{\mathbb{S}}

\newcommand{\EE}{\text{\sf E}}
\newcommand{\Var}{\text{\sf Var}}

\newcommand{\FA}{\text{\sf FA}}

\begin{document}

\begin{frontmatter}



\title{First and Second Moments and Fractional Anisotropy \\
of General von Mises-Fisher and Peanut Distributions}


\author[label1,label2]{Alexandra Shyntar, Thomas Hillen} 


 \affiliation[label1]{organization={University of Alberta},
             addressline={North Campus},
             city={Edmonton},
             postcode={T6G 2G1},
             state={Canada},
             country={shyntar@ualberta.ca}}
 \affiliation[label2]{organization={University of Alberta},
             addressline={North Campus},
             city={Edmonton},
             postcode={T6G 2G1},
             state={Canada},
             country={thillen@ualberta.ca}}

\begin{abstract}
Spherical distributions, in particular, the von Mises-Fisher distribution, are often used for problems using or modelling directional data. Since expectation and variance-covariance matrices follow from the first and second moments of the spherical distribution, the moments often need to be approximated numerically by computing trigonometric integrals. Here, we derive the explicit forms of the first and second moments for an $n$-dimensional von Mises-Fisher and peanut distributions by making use of the divergence theorem in the calculations. The derived formulas can be easily used in simulations, significantly decreasing the computation time. Moreover, we compute the fractional anisotropy formulas for the diffusion tensors derived from the bimodal von Mises-Fisher and peanut distributions, and show that the peanut distribution is limited in the amount of anisotropy it permits, making the von Mises-Fisher distribution a better choice when modelling anisotropy.
\end{abstract}

\begin{keyword}
Spherical distributions \sep Second moment \sep von Mises-Fisher distribution \sep Peanut distribution \sep Fractional anisotropy



\end{keyword}

\end{frontmatter}




\section{Introduction}
Spherical probability distributions have been studied for decades \cite{mardia} and applied in many scientific disciplines, most prominently in geology and astronomy \cite{kunze2004bingham,matheson2023mises}. More recently, spherical distributions have found applications in biology and medicine, whenever the orientation of a moving individual is of importance. This applies to wolves following seismic lines \cite{wolf}, sea turtles finding their breeding islands \cite{turtle}, the migration of whales and butterflies \cite{whale,butterfly}, as well as oriented cell motion such as leukocytes directed motion \cite{leukocyte} and directed invasion of brain cancer cells \cite{swanhillen}.
In all these cases, spherical distributions are used to fit movement data, and to build mathematical models.

Needless to say, the moments of these distributions, in particular expectation and variance, are of central importance, and for some of the distributions these are unknown. We close this gap by computing first and second moments for the general $n$-dimensional unimodal and bimodal von Mises-Fisher distributions, and the $n$-dimensional peanut distribution. We also determine the diffusion tensors in two and three dimensions for these spherical distributions and calculate their eigenvalues and fractional anisotropy formulas.

Mathematically, a spherical distribution is a probability distribution on the unit sphere in $\RR^n$, where $n>0$ denotes the space dimension. We usually denote it as a measure with density $q(\theta)$, $\theta\in \SSS^{n-1}$ with the basic properties
\begin{equation}
\label{q1}  q(\theta)\geq 0, \quad \mbox{for all } \quad \theta\in \SSS^{n-1}, \quad \mbox{ and }\quad  \int_{\SSS^{n-1}} q(\theta) d\theta =1.
\end{equation}
The $n$-dimensional {\bf von Mises-Fisher distribution} is defined for a given unit vector $u\in \SSS^{n-1}$ and a given concentration parameter $k\geq 0$ as  \cite{hillenmoments}
\begin{equation}\label{vonmises}
q(\theta)=\frac{k^{\frac{n}{2}-1}}{(2\pi)^{\frac{n}{2}}I_{\frac{n}{2}-1}(k)}e^{k\theta\cdot u},
\end{equation}
where  $I_p(k)$ with $p\in \mathbb{C}$ defines the modified Bessel function of first kind of order $p$ \cite{toomanyformulas}. If $n=2$ the von Mises-Fisher distribution is simply called the von Mises distribution \cite{hillenmoments}.

The $n$-dimensional {\bf peanut distribution} is defined for a given positive definite anisotropy matrix $A\in \RR^{n\times n}$ as \cite{peanutpaper} 
\begin{equation}\label{peanut}
q(\theta)=\frac{n}{|\mathbb{S}^{n-1}|tr(A)}\theta^T\ A \,\theta
\end{equation}
where $tr(A)$ denotes the trace of $A$, and $|\mathbb{S}^{n-1}|$ is the surface area of an $n-1$ dimensional sphere.

\subsection{Some Basic Identities and Notations}
We denote expectation and variance of the distribution $q(\theta)$ as 
\begin{equation}\label{expectation}
\EE[q] = \int_{\mathbb{S}^{n-1}}^{}\theta q(\theta)d\theta, \qquad 
\Var[q] = \int_{\mathbb{S}^{n-1}}^{} \theta\theta^Tq(\theta)d\theta - \EE[q]\EE[q]^T.
\end{equation}
Here, $ab^T$ for $a,b\in\mathbb{R}^{n}$ denotes the tensor product between two vectors.

The modified Bessel functions of first kind are  defined as \cite{defBessel}
\begin{equation}\label{I}
I_p(x)= \sum_{m=0}^{\infty}\frac{1}{m! \ \Gamma(p+m+1)}\left(\frac{x}{2} \right)^{2m+p},
\end{equation}
where $p\in \mathbb{C} $, $x\in \mathbb{C} $. We will make use the following identities 
\begin{equation}\label{IiJ}
I_p(x)=i^{-p}J_p(ix),\qquad 
\frac{d}{dx}(x^p J_p(x))=x^p J_{p-1}(x),
\end{equation}
where $J_p(x)$ is a Bessel function of order $p$ and $p\in\mathbb{R}$ \cite{toomanyformulas}. We will also make use of the following expansion for large $x\in \RR$, $x>0$ \cite{olver2010nist}
\begin{equation}\label{large_k_bessel}
    I_p(x) \sim \frac{e^x}{\sqrt{2\pi x}}\left(1+ \sum_{n=0}^{\infty}\frac{(-1)^n}{n!(8x)^n} \left( \prod_{m=1}^{n} 4p^2-(2m-1)^2 \right) \right).
\end{equation}
for fractional and integer orders $p$.

We denote by $D_v$ an operator that performs the gradient with respect to $v\in\mathbb{R}^n$ and we note that 
\begin{equation}\label{matrixderivate}
D_x(x^TAx) = Ax+A^Tx
\end{equation}
where $A\in\mathbb{R}^{n\times n}$. 

\section{Von Mises-Fisher Distribution in $n$ dimensions}
We now calculate the first and second moments of the general $n$-dimensional von Mises-Fisher distribution \eqref{vonmises} and apply our results for a bimodal von Mises-Fisher distribution in $n$ dimensions. 

\begin{theorem}\label{thm_vonmises} Consider the 
 von Mises-Fisher distribution \eqref{vonmises} in any dimension $n$. The expectation and variance-covariance matrix are 
 \begin{eqnarray}
 	\EE[q]&=&\frac{I_{\frac{n}{2}}(k)}{I_{\frac{n}{2}-1}(k)}\,u,\label{expectation_vonmises}\\
 \Var[q]&=&\frac{I_{\frac{n}{2}}(k)}{kI_{\frac{n}{2}-1}(k)} \mathbb{I}+\left[\frac{I_{\frac{n}{2}+1}(k)}{I_{\frac{n}{2}-1}(k)} -\left( \frac{I_{\frac{n}{2}}(k)}{I_{\frac{n}{2}-1}(k)} \right)^2\right]uu^T. \label{variance_vonmises}
 \end{eqnarray}
\end{theorem}

\begin{proof}
{\bf Expectation:} Let
\begin{equation}\label{cvonmises}
	c= \frac{k^{\frac{n}{2}-1}}{(2\pi)^{\frac{n}{2}}I_{\frac{n}{2}-1}(k)}.
	\end{equation}
Let $b\in\mathbb{R}^{n}$ be an arbitrary vector, then using \eqref{vonmises} we obtain
	\begin{equation}
	\begin{aligned}
	b\cdot \frac{1}{c} \EE[q] & = b\cdot\int_{\mathbb{S}^{n-1}}^{} \theta e^{k\theta u}d\theta 
	= \int_{\mathbb{B}_1(0)}^{} \frac{\partial}{\partial v^j}b_j e^{kv_l u^l}dv
	= b\cdot uk\int_{\mathbb{B}_1(0)}^{}e^{kvu}dv\\
	&= b\cdot uk\int_{0}^{1}\int_{\mathbb{S}^{n-1}}^{}e^{krvu}r^{n-1}dvdr
	= b\cdot u\frac{(2\pi)^{\frac{n}{2}}}{k^{\frac{n}{2}-2}}\int_{0}^{1}r^{\frac{n}{2}}I_{\frac{n}{2}-1}(rk)dr,
    \label{calcs333}
	\end{aligned}
	\end{equation}
	where we use Einstein summation convention for repeated indices, we applied the divergence theorem at the second equivalence and applied \eqref{q1} at the last equivalence.
	
	Making use of the Bessel identities \eqref{IiJ}, we can rewrite the term within the integral, that is
	\begin{equation}\label{maybehandy}
	\begin{aligned}
	r^{\frac{n}{2}}I_{\frac{n}{2}-1}(rk)=r^{\frac{n}{2}}i^{-\frac{n}{2}+1}J_{\frac{n}{2}-1}(irk)=\frac{i^{-n}}{k^{\frac{n}{2}+1}}\frac{d}{dr}(irk)^{\frac{n}{2}}J_{\frac{n}{2}}(irk).
	\end{aligned}
	\end{equation}
	Using \eqref{maybehandy}, we can simplify further and apply the fundamental theorem of calculus to obtain that
	\begin{equation}\label{eq333}
	b\cdot u\frac{(2\pi)^{\frac{n}{2}}}{k^{\frac{n}{2}-2}}\int_{0}^{1}r^{\frac{n}{2}}I_{\frac{n}{2}-1}(rk)dr=b\cdot u\frac{(2\pi)^{\frac{n}{2}}}{k^{\frac{n}{2}-1}}i^{-\frac{n}{2}}J_{\frac{n}{2}}(ik)=b\cdot u\frac{(2\pi)^{\frac{n}{2}}}{k^{\frac{n}{2}-1}}I_{\frac{n}{2}}(k).
	\end{equation}
	Therefore, we find that 
	\begin{equation*}
	b\cdot \frac{1}{c} \EE[q] = b\cdot u\frac{(2\pi)^{\frac{n}{2}}}{k^{\frac{n}{2}-1}}I_{\frac{n}{2}}(k).
	\end{equation*}
	After simplifying, and noting that $b$ is an arbitrary vector, it follows that
	\begin{equation}\label{vMexpectation}
	\EE[q]=\frac{I_{\frac{n}{2}}(k)}{I_{\frac{n}{2}-1}(k)}u.
	\end{equation}

\noindent {\bf Second moment:}
We now calculate the second moment as well as the variance-covariance matrix for the von Mises-Fisher distribution \eqref{vonmises}. We directly see that 
\begin{equation}\label{form1}
\EE[q]\EE[q]^T=\left( \frac{I_{\frac{n}{2}}(k)}{I_{\frac{n}{2}-1}(k)} \right)^2uu^T.
\end{equation}
To compute the second moment, we take two test vectors  $a,b\in\mathbb{R}^{n}$, and compute 
\begin{equation}
\begin{aligned}\notag
a^T\int_{\mathbb{S}^{n-1}}^{} \theta\theta^Tq(\theta)d\theta \, b 
&= c\int_{\mathbb{B}_1(0)}^{}\frac{\partial}{\partial v^j}(a_jb_mv^me^{kv^l u_l})dv\\
&=c\left(a^Tb \int_{\mathbb{B}_1(0)}^{}e^{kvu}dv+ka^Tub\cdot\int_{\mathbb{B}_1(0)}^{}ve^{kvu}dv \right).
\end{aligned}
\end{equation}

We already know that 
$
\int_{\mathbb{B}_1(0)}^{}e^{kvu}dv=\frac{(2\pi)^{\frac{n}{2}}I_{\frac{n}{2}}(k)}{k^{\frac{n}{2}}},
$
from \eqref{calcs333}  and \eqref{eq333}. Further, 
\begin{equation}
\begin{aligned}\notag
\int_{\mathbb{B}_1(0)}^{}ve^{kvu}dv 
&=\int_{0}^{1} r^n\int_{\mathbb{S}^{n-1}}^{}ve^{kr\theta u}dvdr=\frac{(2\pi)^{\frac{n}{2}}}{k^{\frac{n}{2}-1}}u \int_{0}^{1} r^{\frac{n}{2}+1}I_{\frac{n}{2}}(rk)dr
\end{aligned}
\end{equation}
by applying \eqref{q1} at the second equivalence. By using the Bessel identities \eqref{IiJ}, we can rewrite the term in the integral as
\begin{equation}\label{handy}
r^{\frac{n}{2}+1}I_{\frac{n}{2}}(rk)=r^{\frac{n}{2}+1}i^{-\frac{n}{2}}J_{\frac{n}{2}}(irk)=\frac{i^{-n-2}}{k^{\frac{n}{2}+2}}\frac{d}{dr}(irk)^{\frac{n}{2}+1}J_{\frac{n}{2}+1}(irk).
\end{equation}
Using \eqref{handy} and \eqref{IiJ}, we obtain that
\begin{equation}
\frac{(2\pi)^{\frac{n}{2}}}{k^{\frac{n}{2}-1}}u \int_{0}^{1} r^{\frac{n}{2}+1}I_{\frac{n}{2}}(rk)dr =\frac{(2\pi)^{\frac{n}{2}}i^{-\frac{n}{2}-1}}{k^{\frac{n}{2}}}uJ_{\frac{n}{2}+1}(ik)=\frac{(2\pi)^{\frac{n}{2}}}{k^{\frac{n}{2}}}uI_{\frac{n}{2}+1}(k).
\end{equation}
Thus,
\begin{equation}\label{int2}
\int_{\mathbb{B}_1(0)}^{}ve^{kvu}dv = \frac{(2\pi)^{\frac{n}{2}}}{k^{\frac{n}{2}}}uI_{\frac{n}{2}+1}(k).
\end{equation}
Using \eqref{eq333} and \eqref{int2}, we find that
\begin{equation}\label{form2}
\int_{\mathbb{S}^{n-1}}^{} \theta\theta^Tq(\theta)d\theta = \frac{I_{\frac{n}{2}}}{kI_{\frac{n}{2}-1}} \mathbb{I}+\frac{I_{\frac{n}{2}+1}}{I_{\frac{n}{2}-1}}uu^T,
\end{equation}
since $a,b$ are arbitrary.

Substituting \eqref{form1} and \eqref{form2}  into the variance formula \eqref{expectation}, we obtain that
\begin{equation}
\Var[q]=\frac{I_{\frac{n}{2}}(k)}{kI_{\frac{n}{2}-1}(k)} \mathbb{I}+\frac{I_{\frac{n}{2}+1}(k)}{I_{\frac{n}{2}-1}(k)}uu^T -\left( \frac{I_{\frac{n}{2}}(k)}{I_{\frac{n}{2}-1}(k)} \right)^2uu^T.
\end{equation}
\end{proof}

\begin{lemma}
Consider the bimodal von Mises-Fisher distribution
\begin{equation}\label{bi_vonmises}
q(\theta)=\frac{k^{\frac{n}{2}-1}}{2(2\pi)^{\frac{n}{2}}I_{\frac{n}{2}-1}(k)}(e^{k\theta\cdot u}+e^{-k\theta\cdot u}).
\end{equation}
The expectation and the variance-covariance matrix for \eqref{bi_vonmises} is given by
 \begin{eqnarray*}
 	\EE[q]&=&0,\\
 \Var[q]&=&\frac{I_{\frac{n}{2}}(k)}{kI_{\frac{n}{2}-1}(k)} \mathbb{I}+\frac{I_{\frac{n}{2}+1}(k)}{I_{\frac{n}{2}-1}(k)} uu^T.
 \end{eqnarray*}
\end{lemma}
\begin{proof}
These formulas follow easily from the formulas determined in Theorem \ref{thm_vonmises} by direct summation of the corresponding terms.

\end{proof}

{\bf Example for $n=3$:} It is interesting to consider the important case of the von Mises-Fisher distribution (\ref{vonmises}) for $n=3$. For $n=3$ we find from Theorem \ref{thm_vonmises} that 
 \begin{equation*}
 	\EE[q]=\frac{I_{\frac{3}{2}}(k)}{I_{\frac{1}{2}}(k)}\,u, \qquad 
 \Var[q]=\frac{I_{\frac{3}{2}}(k)}{kI_{\frac{1}{2}}(k)} \mathbb{I}+\left[\frac{I_{\frac{5}{2}}(k)}{I_{\frac{1}{2}}(k)} -\left( \frac{I_{\frac{3}{2}}(k)}{I_{\frac{1}{2}}(k)} \right)^2\right]uu^T.
 \end{equation*}
In \cite{hillenmoments}, these terms were represented as 
\begin{equation*}
\EE[q] = \left(\coth k - \frac{1}{k}\right)u, \qquad 
\Var[q] = \left( \frac{\coth k}{k} - \frac{1}{k^2} \right) \mathbb{I} + \left[ 1-\frac{\coth k}{k} + \frac{2}{k^2} - \coth^2k \right]uu^T.
\end{equation*}
Hence we can compare the coefficients and find the following identities for Bessel functions:
\begin{eqnarray}
\frac{I_{\frac{3}{2}}(k)}{I_{\frac{1}{2}}(k)} &=&\coth k - \frac{1}{k}\label{bessel_formula2},\\
\frac{I_{\frac{5}{2}}(k)}{I_{\frac{1}{2}}(k)} -\left( \frac{I_{\frac{3}{2}}(k)}{I_{\frac{1}{2}}(k)} \right)^2 &=&  1-\frac{\coth k}{k} + \frac{2}{k^2} - \coth^2k.\label{bessel_formula3}
\end{eqnarray}
We emphasize that these formulas \eqref{bessel_formula2} and \eqref{bessel_formula3} are not new and appear in different forms in \cite{toomanyformulas}.

\section{Peanut distribution in $n$ dimensions}

We first note that the peanut distribution \eqref{peanut} is of the form 
\[ q(\theta)= F(\theta^T A \theta),\]
with an appropriate function $F$, where $A$ is a given matrix. 

There are two other prominent distributions that have the form $F(\theta^T A^{-1} \theta)$, which is the $n$-dimensional {\bf ordinary distribution function (ODF)} \cite{conte} 
\begin{equation}\label{ODF}
q(\theta)=\frac{1}{4\pi|A|^{\frac{1}{2}}(\theta^T A^{-1}\theta)^{\frac{3}{2}}},
\end{equation}
and  the {\bf Bingham distribution}  \cite{mardia, bingham} 
\begin{equation}\label{bingham}
q(\pm \theta) =   \frac{1}{\sqrt{|A|(4\pi\Delta)^3}}e^{\frac{-\theta^T A^{-1}\, \theta}{4\Delta}}
\end{equation}
where $|A|$ denotes the determinant of $A$ and $\Delta$ is the diffusion time. The Bingham distribution is also known as anisotropic multivariate normal distribution. 

In the case of ODF and the Bingham distributions we actually use $A^{-1}$, but the structure $F(\theta^T A \theta)$ is the same. Due to symmetry in $\theta$, the first moment of these distributions equals zero:
\begin{lemma} \label{Ezero} If $A$ is a given matrix and $q$ is of the form $q(\theta)=F(\theta^T A \theta)$ or $q(\theta)=F(\theta^T A^{-1} \theta)$ then  
$\EE[q]=0$. Moreover, all odd moments are zero too.  
\end{lemma}
\begin{proof}
It is sufficient to consider the case $q(\theta)=F(\theta^T A \theta)$. Let $\theta_1$ denote the first coordinate of $\theta$ in $\RR^n$. We split the sphere $\SSS^{n-1}$ according to positive and negative $\theta_1$ component:
\[ S^+ := \{\theta\in \SSS^{n-1}: \theta_1\geq 0\} \mbox { and } S^- := \{\theta\in \SSS^{n-1}: \theta_1< 0\}.\]
Then 
\begin{eqnarray*}
\EE[q] & = &  \int_{S^+} \theta F(\theta^T A \theta) d\theta + \int_{S^-}  \theta F(\theta^T A \theta) d\theta \\
&=&  \int_{S^+} \theta F(\theta^T A \theta) d\theta + \int_{-S^+}  -\theta F((-\theta^T) A (-\theta)) d(-\theta)\\
&=& 0. 
\end{eqnarray*}
\end{proof}

\noindent For the second moment we focus on the $n$-dimensional peanut distribution
\begin{theorem}\label{var_peanut}
Consider the peanut distribution \eqref{peanut}. Then
\begin{eqnarray*}
\EE[q] &=& 0, \\    
\Var[q]  &=& \frac{1}{(n+2)}\mathbb{I}+\frac{1}{(n+2)tr(A)}(A+A^T).
\end{eqnarray*}
\end{theorem}
\begin{proof} We have seen in Lemma \ref{Ezero} that $\EE(q)=0$. For simplicity, we define the coefficient in \eqref{peanut} as
\begin{equation}\label{cpeanut}
c=\frac{n}{|\mathbb{S}^{n-1}|tr(A)}.
\end{equation} 
To compute the second moment, 
we multiply $\Var[q]$ by two arbitrary vectors $a,b\in\mathbb{R}^{n}$ in order to obtain a scalar value, and then apply the divergence theorem. Using \eqref{peanut} and \eqref{cpeanut}, we obtain the following result
\begin{equation}
\begin{aligned}\notag
a^T (\Var[q] )b &= \int_{\mathbb{S}^{n-1}}^{}a^T\theta \theta^Tq(\theta)bd\theta 
=c\int_{\mathbb{B}_1(0)}^{}\frac{\partial}{\partial v_i}(a_ib_jv^jv^TAv )dv\\
&= a^T\left(c\int_{\mathbb{B}_1(0)}^{}v^TA vdv\ \right)b+a^T\left(c\int_{\mathbb{B}_1(0)}^{}(A+A^T)vv^Tdv \right) b,
\end{aligned}
\end{equation}
where in the second equivalence divergence theorem was applied and in the last equivalence the product rule was used.

Next, we simplify the final two terms in the above calculation. The first term simplifies as
\begin{equation}
\begin{aligned}
	a^T\left(c\int_{\mathbb{B}_1(0)}^{}v^TA vdv\ \right)b&=a^T\left(c\int_{0}^{1}r^{n+1}\int_{\mathbb{S}^{n-1}}^{}\theta^TA\theta d\theta\right) b= \frac{a^Tb}{n+2}, 
\end{aligned}
\end{equation}
since $\int_{\mathbb{S}^{n-1}}^{}q(\theta)d\theta=1$, implying that $\int_{\mathbb{S}^{n-1}}^{}\theta^TA\theta d\theta=1/c.$ For the second term, we set $J=A+A^T$ and simplify
\begin{eqnarray*}
&&a^T\left(c\int_{\mathbb{B}_1(0)}^{}Jvv^Tdv \right) b
=c\int_{0}^{1}r^{n+1}dr\int_{\mathbb{S}^{n-1}}^{}\theta_ja^iJ_{ij}\theta^mb_md\theta\\
&&= \frac{c}{n+2}\int_{\mathbb{B}_1(0)}^{}\frac{\partial}{\partial w_j}(a^iJ_{ij}w^mb_m)dw
=\frac{c}{n+2}a^iJ_{ij}b_j|\mathbb{B}_1(0)|=\frac{1}{(n+2)tr(A)}a^TJb.
\end{eqnarray*}
In the above calculation, divergence theorem was applied at the second equivalence and the last equivalence made use of $\frac{n|\mathbb{B}_1(0)|}{|\mathbb{S}^{n-1}|}=1$ \cite{sphere_ball_vols}.

Therefore, since $a,b\in \mathbb{R}^n$ are arbitrary, it follows that
\begin{equation}
\Var[q]  =\frac{1}{(n+2)}\mathbb{I}+\frac{1}{(n+2)tr(A)}(A+A^T).
\end{equation}
\end{proof}

\section{Eigenvalues and Anisotropy}
When spherical distributions are used in diffusion models, then the variance-covariance matrix describes the anisotropic movement of the random walkers \cite{wolf}. The level of anisotropy can be measured by the eigenvalues of the diffusion tensor and by the {\it fractional anisotropy} ($\FA$) \cite{swanhillen}. 
The fractional anisotropy $\FA$ is an index between $0$ and $1$, where $0$ is the case of full radial symmetry giving isotropy and $1$ is the case of extreme anisotropic alignment to a particular direction.  Another indicator of anisotropy is the ratio $R$ of the largest to the smallest eigenvalue of the diffusion tensor. This quantity ranges between $1$ (isotropic) to $\infty$ (maximal anisotropic). We begin with the definition of the diffusion tensor, based on velocity jump processes, as derived in many publications \cite{hillenmoments,wolf,swanhillen,Engwer}. 

\begin{definition}
	The diffusion tensor $\mathbb{D}$ of a general velocity-jump process is proportional to the variance-covariance matrix $\Var[q]$ of a given spherical distribution $q$ \cite{hillenmoments}. It is defined by
	\begin{equation}
	\mathbb{D}(t,x) = \frac{s^2}{\mu}\Var[q],
	\label{diff_tensor_def}
	\end{equation}
	where $s$ is the speed and $\mu$ is the turning rate of a cell or individual \cite{hillenmoments}.
\end{definition}

The fractional anisotropy (FA) formulas are defined separately for two and three dimensions and we summarize them in the following definition. 

\begin{definition}
\begin{enumerate}
\item	The fractional anisotropy in two dimensions is given by
	\begin{equation}
	\FA_2 = \sqrt{\frac{2[(\lambda_1-\bar{\lambda})^2+(\lambda_2-\bar{\lambda})^2]}{\lambda_1^2+\lambda_2^2}}
	\label{FA2}
	\end{equation} 
	where $\lambda_1,\lambda_2$ are the eigenvalues of a two dimensional tensor and $\bar{\lambda}$ is the average of the two eigenvalues \cite{swanhillen}.
\item 	
	The fractional anisotropy in three dimensions is given by
	\begin{equation}
	\FA_3 = \sqrt{\frac{3[(\lambda_1-\bar{\lambda})^2+(\lambda_2-\bar{\lambda})^2+(\lambda_3-\bar{\lambda})^2]}{2(\lambda_1^2+\lambda_2^2+\lambda_3^2)}}
	\label{FA3}
	\end{equation} 
	where $\lambda_1,\lambda_2,\lambda_3$ are the eigenvalues of a three dimensional tensor and $\bar{\lambda}$ is the average of the three eigenvalues \cite{swanhillen}.
    \item The {\it anisotropy ratio} $R$ is defined as 
    \begin{equation}\label{R}
    R = \frac{\max\{\mbox{eigenvalues of } \mathbb{D}\}}{\min\{\mbox{eigenvalues of } \mathbb{D}\}}
    \end{equation}

\end{enumerate}

\end{definition}

\begin{lemma}\label{eig_peanut}
Given a symmetric, positive definite matrix $A$ with eigenvalue and eigenvector pairs $(\hat\lambda_i, \hat v_i)$, $i=1,\dots,n$. The diffusion tensor arising from the peanut distribution \eqref{peanut} is
\begin{equation}
	\mathbb{D} = \frac{s^2}{\mu(n+2)}\mathbb{I}+\frac{2s^2}{\mu(n+2)tr(A)}A.
	\label{diff_tensor_peanut}
\end{equation}
where $A$ is a symmetric, positive definite matrix. The diffusion tensor \eqref{diff_tensor_peanut}  has the eigenvalues
$$\lambda_n = \frac{s^2}{\mu(n+2)}\left(1+\frac{2\hat{\lambda}_n}{tr(A)}  \right).$$
The fractional anisotropies and the ratio of eigenvalues are bounded as 
\begin{eqnarray}
0 \ \ \leq  \ \ \FA_2 & =& 2|\hat{\lambda}_1-\hat{\lambda}_2|\sqrt{\frac{1}{(tr(A)+2\hat{\lambda}_1)^2+(tr(A)+2\hat{\lambda}_2)^2}} \ \ \leq \ \  \frac{2}{\sqrt{10}}\label{FA2_peanut}\\
0 \ \ \leq \ \ \FA_3 &=&\sqrt{\frac{2[(2\hat{\lambda}_1+\hat{\lambda}_2-\hat{\lambda}_3)^2+(2\hat{\lambda}_2-\hat{\lambda}_1-\hat{\lambda}_3)^2+(2\hat{\lambda}_3-\hat{\lambda}_1-\hat{\lambda}_2)^2]}{3[(tr(A)+2\hat{\lambda}_1)^2+(tr(A)+2\hat{\lambda}_2)^2+(tr(A)+2\hat{\lambda}_3)^2]}} \ \ \leq \ \ \frac{2}{\sqrt{11}}\label{FA3_peanut}\\
1 \ \ \leq \ \ R &=&\frac{tr(A)+2max(\hat{\lambda}_n)}{tr(A)+2min(\hat{\lambda}_n)} \ \ \leq \ \ 3
\end{eqnarray}
\end{lemma}
\begin{proof}

	To compute the eigenvalues, we multiplying both sides of \eqref{diff_tensor_peanut} by the eigenvector $\hat v_n$ of $A$  and obtain
	$$\mathbb{D}\hat{v}_n = \frac{s^2}{\mu(n+2)}\hat{v}_n+\frac{2s^2}{\mu(n+2)tr(A)}\hat{\lambda}_n\hat{v}_n=\frac{s^2}{\mu(n+2)}\left(1+\frac{2\hat{\lambda}_n}{tr(A)}  \right) \hat{v}_n.$$

  The formulas for the fractional anisotropy follow directly by substituting the eigenvalues of \eqref{diff_tensor_peanut}, and  to the upper bound follows from taking $\lambda_1\rightarrow \infty$. For the ratio $R$ we write 
    \[
    R = \frac{\hat \lambda_1+\sum_{i=2}^{n-1} \hat \lambda_i + 3 \hat \lambda_n}{3\hat \lambda_1+\sum_{i=2}^{n-1}\hat \lambda_i + \hat \lambda_n}
    \leq  \frac{9\hat \lambda_1 +3\sum_{i=2}^{n-1} \hat \lambda_i + 3 \hat \lambda_n}{3\hat \lambda_i+\sum_{i=2}^{n-1}\hat \lambda_i + \hat \lambda_n}\\
    \leq 3.    
\]
where the maximum eigenvalue and minimum eigenvalues of $A$ are denoted by $\hat \lambda_n$ and $\hat \lambda_1$, respectively. All other eigenvalues are called $\hat \lambda_i$ and the trace becomes the sum of the eigenvalues.   
\end{proof}

It is very surprising to see that the anisotropy ratio is bounded by a finite value of $3$ in any space dimension, and the fractional anisotropies  are uniformly bounded by a value less than 1. This means the peanut distribution does not allow for arbitrary anisotropies, and for that reason it is of limited use in the modelling of biological movement data. 

Next, we consider the anisotropy of the bimodal von Mises-Fisher distribution \eqref{bi_vonmises}.
\begin{theorem}\label{eig_bi_vonmises}
Let $u\in \SSS^{n-1}$ be a given direction and $k$ the concentration parameter. The diffusion tensor arising from the bimodal von Mises-Fisher distribution \eqref{bi_vonmises} is of the form
\begin{equation}
	\mathbb{D} = \alpha(k) \mathbb{I} + \beta(k) \, u u^T,
    \quad \mbox{where}\quad \alpha(k)=\frac{s^2I_{\frac{n}{2}}(k)}{\mu kI_{\frac{n}{2}-1}(k)}\quad \mbox{and} \quad \beta(k) = \frac{s^2 I_{\frac{n}{2}+1}(k)}{\mu I_{\frac{n}{2}-1}(k)}.
    	\label{diff_tensor_bi_vonmises}
\end{equation}
 The diffusion tensor \eqref{diff_tensor_bi_vonmises} has eigenvalues 
	$$\lambda_1 = \alpha+\beta, \quad \lambda_2 \ ... \  \lambda_n = \alpha. $$
The anisotropies and the ratio of the eigenvalues satisfy 
\begin{eqnarray*}
0\leq \FA_2 &=& |\beta|\sqrt{\frac{1}{(\alpha+\beta)^2+\alpha^2}} = \sqrt{\frac{1}{(\frac{\alpha}{\beta}+1)^2+(\frac{\alpha}{\beta})^2}}\leq 1\\
   0\leq \FA_3 &=& |\beta| \sqrt{\frac{1}{(\alpha+\beta)^2+2\alpha^2}} = \sqrt{\frac{1}{(\frac{\alpha}{\beta}+1)^2+2(\frac{\alpha}{\beta})^2}}\leq 1\\
        1\leq R &=& 1+\frac{\beta}{\alpha} \leq \infty.
    \end{eqnarray*}
\end{theorem}

\begin{proof}
To compute the eigenvalues for \eqref{diff_tensor_bi_vonmises} we first take the vector $u$ and show that it is indeed an eigenvector of \eqref{diff_tensor_bi_vonmises} since
	\begin{equation}
	\mathbb{D} u = \alpha u+ \beta uu^Tu=(\alpha+\beta)u,
	\end{equation}
	hence, $\lambda_1 = \alpha+\beta$. Now we take any $u^\perp\in \mathbb{S}^{n-1}$ (where $u^Tu^\perp=0$) and show that it is an eigenvector of \eqref{diff_tensor_bi_vonmises} as well since
	\begin{equation}
	\mathbb{D} u^\perp = \alpha u^\perp+ \beta uu^Tu^\perp=\alpha u^\perp.
	\end{equation}
	So $u^\perp\in \mathbb{S}^{n-1}$ is an eigenvector with an eigenvalue $\lambda = \alpha$. Since we can define $n-1$ perpendicular vectors, we conclude that $\lambda_2 \ ... \  \lambda_n = \alpha.$

    The formulas for the anisotropies and for $R$ follow from substituting the determined eigenvalues into the definitions. To determine the upper bounds for the formulas we compute the limits of $\alpha(k)$ and $\beta(k)$. We start with the isotropic case of $k \rightarrow 0$. We apply the expansion of Bessel functions as shown in \eqref{I}
\begin{equation} \label{alpha_to_0(k)}
    \begin{aligned}
        \lim_{k\to0} \alpha(k)  &=\frac{s^2}{\mu} \lim_{k\to0}\frac{\frac{1}{\Gamma(\frac{n}{2}+1)2^{\frac{n}{2}}}k^{\frac{n}{2}}+ \frac{1}{\Gamma(\frac{n}{2}+2)2^{\frac{n}{2}+2}}k^{\frac{n}{2}+2}+h.o.t}{k\left[ \frac{1}{\Gamma(\frac{n}{2})2^{\frac{n}{2}-1}}k^{\frac{n}{2}-1}+\frac{1}{\Gamma(\frac{n}{2}+1)2^{\frac{n}{2}+1}}k^{\frac{n}{2}+1}+h.o.t\right]}
        &= \frac{s^2}{\mu n} .
    \end{aligned}
\end{equation}
and 
\begin{equation}\label{beta_to_0}
    \begin{aligned}
         \lim_{k\to0} \beta(k) &= \frac{s^2}{\mu} \lim_{k\to0}\frac{\frac{1}{\Gamma(\frac{n}{2}+2)2^{\frac{n}{2}+1}}k^{\frac{n}{2}+1}+ \frac{1}{\Gamma(\frac{n}{2}+3)2^{\frac{n}{2}+3}}k^{\frac{n}{2}+3}+h.o.t}   {\frac{1}{\Gamma(\frac{n}{2})2^{\frac{n}{2}-1}}k^{\frac{n}{2}-1}+\frac{1}{\Gamma(\frac{n}{2}+1)2^{\frac{n}{2}+1}}k^{\frac{n}{2}+1}+h.o.t}
        &=0.
    \end{aligned}
\end{equation}
To evaluate the limit  as $k\rightarrow \infty$ we apply the expansion \eqref{large_k_bessel} for large $k$ to obtain
\begin{equation}
\begin{aligned}
    \lim_{k\to \infty} \alpha(k) &=\frac{s^2}{\mu} \lim_{k\to \infty} \frac{\frac{e^k}{\sqrt{2\pi k}}[1 - \frac{1}{8k}(4(\frac{n}{2})^2-1)+ \ ...] }{k\frac{e^k}{\sqrt{2\pi k}}[1 - \frac{1}{8k}(4(\frac{n}{2}-1)^2-1) + \ ...]}=0,
\end{aligned}
\end{equation}
and 
\begin{equation}
    \begin{aligned}
        \lim_{k\to \infty} \beta(k) &=\frac{s^2}{\mu} \lim_{k\to \infty}  \frac{\frac{e^k}{\sqrt{2\pi k}}[1 - \frac{1}{8k}(4(\frac{n}{2}+1)^2-1)+ \ ...] }{\frac{e^k}{\sqrt{2\pi k}}[1 - \frac{1}{8k}(4(\frac{n}{2}-1)^2-1) + \ ...]}=1.
    \end{aligned}
\end{equation}

Now we compute the bounds for the ratio $R(k)$. 
We see that if $k\rightarrow 0$, $R(k) \rightarrow 1$ by the above calculations. Further, $R(k) \rightarrow \infty$ as $k\rightarrow \infty$ since
\begin{equation}
    \lim_{k\to \infty} \frac{\beta(k)}{\alpha(k)} = \lim_{k\to \infty} \frac{kI_{\frac{n}{2}+1}}{I_{\frac{n}{2}}} =\lim_{k\to \infty}  \frac{k\frac{e^k}{\sqrt{2\pi k}}[1 - \frac{1}{8k}(4(\frac{n}{2}+1)^2-1)+ \ ...] }{\frac{e^k}{\sqrt{2\pi k}}[1 - \frac{1}{8k}(4(\frac{n}{2})^2-1) + \ ...]}=\infty.
\end{equation}
And finally,
\begin{equation}
    \lim_{k\to \infty} \frac{\alpha(k)}{\beta(k)} =  \lim_{k\to \infty} \frac{I_\frac{n}{2}}{kI_{\frac{n}{2}+1}} = \lim_{k\to \infty} \frac{\frac{e^k}{\sqrt{2\pi k}}[1 - \frac{1}{8k}(4(\frac{n}{2})^2-1) + \ ...]}{k\frac{e^k}{\sqrt{2\pi k}}[1 - \frac{1}{8k}(4(\frac{n}{2}+1)^2-1)+ \ ...]}=0.
\end{equation}
So,
 $\FA_2\rightarrow 0$ and $\FA_3 \rightarrow 0$ as $k\rightarrow 0$. Further, $\FA_2 \rightarrow 1$ and $\FA_3 \rightarrow 1$ as $k\rightarrow \infty$.  
 \end{proof}

\section{Discussion}

The standard way to compute moments of spherical distributions would be to use polar coordinates and engage is a series of trigonometric integrations. While this is often possible in two and three dimensions, it gets tedious, or impossible, in higher dimensions. Here we present an alternative method that can be used in any space dimension. This method was first used in \cite{hillenmoments} in two and three dimensions. It uses the divergence theorem on the unit ball to transform and simplify the integrals as much as possible. As a result, we obtain explicit formulas for the variance-covariance matrices of the $n$-dimensional von Mises-Fisher distribution (Theorem \ref{thm_vonmises}) and the $n$-dimensional peanut distribution (Theorem \ref{var_peanut}) as new results. We also used this method to consider other spherical distributions such as the ordinary distribution function (ODF) \eqref{ODF} and the Bingham distribution \eqref{bingham}, which, unfortunately did not lead to closed form expressions. 

Once the variance-covariance matrices were established, we considered levels of anisotropy as measured through the fractional anisotropies and the ratio of maximal over minimal eigenvalues. Surprisingly, the anisotropy of the peanut distribution is bounded away from 1 (Lemma \ref{eig_peanut}) making it problematic to use in modelling. The anisotropy of the von Mises-Fisher distribution covers the entire range from isotropic to fully anisotropic (Theorem \ref{eig_bi_vonmises}), and seems to have no restriction if used for modelling. 

 Using the explicit formulas derived here,  will significantly decrease the computation time in problems having a large amount of data that require a higher dimensional von Mises-Fisher distribution  \cite{banerjee2005clustering} and decrease the computation time when simulating the transport equations (partial differential equations) that utilize the biological movement data \cite{butterfly, swanhillen,turtle,wolf}. Moreover, the formulas derived for the generalized von Mises-Fisher distribution can be easily extended to compute the expectation and variance covariance matrices of mixed von Mises-Fisher distributions by substituting the appropriate directional vectors and adding up the terms. The method of integration, developed here, might also inspire researchers to crack some of the unknown integrals that relate to spherical distributions. 

\section*{Acknowledgments}
AS acknowledges the funding from the Alberta Graduate Excellence Scholarship. TH is supported through a discovery grant of the Natural Science and Engineering Research Council of Canada (NSERC), RGPIN-2023-04269.

\section*{Conflict of interest}
All authors declare no conflicts of interest in this paper.






\begin{thebibliography}{10}

\bibitem{toomanyformulas}
{\sc Abramowitz, M., and Stegun, I.~A.}
\newblock {\em Handbook of mathematical functions with formulas, graphs, and mathematical tables}, vol.~55.
\newblock US Government printing office, 1968.

\bibitem{banerjee2005clustering}
{\sc Banerjee, A., Dhillon, I.~S., Ghosh, J., Sra, S., and Ridgeway, G.}
\newblock Clustering on the unit hypersphere using von mises-fisher distributions.
\newblock {\em Journal of Machine Learning Research 6}, 9 (2005).

\bibitem{defBessel}
{\sc Baricz, {\'A}.}
\newblock {\em Generalized Bessel functions of the first kind}.
\newblock Springer, 2010.

\bibitem{bingham}
{\sc Basser, P.~J.}
\newblock Diffusion and diffusion tensor mr imaging: fundamentals.
\newblock {\em Magnetic resonance imaging of the brain and spine\/} (2008), 1752--1767.

\bibitem{whale}
{\sc Chatterjee, D., and Ghosh, P.}
\newblock On spatio-temporal directional association of blue and humpback whale’s migratory path navigation with sun, moon, ocean current \& earth’s magnetic field.
\newblock {\em ESS Open Archive eprints 88\/} (2024), 08896939.

\bibitem{conte}
{\sc Conte, M., and Surulescu, C.}
\newblock Mathematical modeling of glioma invasion: acid-and vasculature mediated go-or-grow dichotomy and the influence of tissue anisotropy.
\newblock {\em Applied Mathematics and Computation 407\/} (2021), 126305.

\bibitem{Engwer}
{\sc Engwer, C., Hillen, T., Knappitsch, M., and Surulescu, C.}
\newblock Glioma follow white matter tracts: a multiscale dti-based model.
\newblock {\em Journal of mathematical biology 71\/} (2015), 551--582.

\bibitem{sphere_ball_vols}
{\sc Folland, G.~B.}
\newblock How to integrate a polynomial over a sphere.
\newblock {\em The American Mathematical Monthly 108}, 5 (2001), 446--448.

\bibitem{wolf}
{\sc Hillen, T., and Painter, K.~J.}
\newblock Transport and anisotropic diffusion models for movement in oriented habitats.
\newblock In {\em Dispersal, individual movement and spatial ecology: A mathematical perspective}. Springer, 2013, pp.~177--222.

\bibitem{hillenmoments}
{\sc Hillen, T., Painter, K.~J., Swan, A.~C., and Murtha, A.~D.}
\newblock Moments of von mises and fisher distributions and applications.
\newblock {\em Mathematical Biosciences and Engineering 14}, 3 (2017), 673--694.

\bibitem{kunze2004bingham}
{\sc Kunze, K., and Schaeben, H.}
\newblock The bingham distribution of quaternions and its spherical radon transform in texture analysis.
\newblock {\em Mathematical Geology 36}, 8 (2004), 917--943.

\bibitem{mardia}
{\sc Mardia, K.~V., and Jupp, P.~E.}
\newblock {\em Directional statistics}.
\newblock John Wiley \& Sons, 2009.

\bibitem{matheson2023mises}
{\sc Matheson, I.~C., Malhotra, R., and Keane, J.~T.}
\newblock A von mises--fisher distribution for the orbital poles of the plutinos.
\newblock {\em Monthly Notices of the Royal Astronomical Society 522}, 3 (2023), 3298--3307.

\bibitem{olver2010nist}
{\sc Olver, F.~W.}
\newblock {\em NIST handbook of mathematical functions hardback and CD-ROM}.
\newblock Cambridge university press, 2010.

\bibitem{peanutpaper}
{\sc Painter, K., and Hillen, T.}
\newblock Mathematical modelling of glioma growth: the use of diffusion tensor imaging (dti) data to predict the anisotropic pathways of cancer invasion.
\newblock {\em Journal of Theoretical Biology 323\/} (2013), 25--39.

\bibitem{butterfly}
{\sc Painter, K.~J.}
\newblock Multiscale models for movement in oriented environments and their application to hilltopping in butterflies.
\newblock {\em Theoretical ecology 7\/} (2014), 53--75.

\bibitem{turtle}
{\sc Painter, K.~J., and Hillen, T.}
\newblock Navigating the flow: individual and continuum models for homing in flowing environments.
\newblock {\em Journal of the Royal Society Interface 12}, 112 (2015), 20150647.

\bibitem{swanhillen}
{\sc Swan, A., Hillen, T., Bowman, J.~C., and Murtha, A.~D.}
\newblock A patient-specific anisotropic diffusion model for brain tumour spread.
\newblock {\em Bulletin of Mathematical Biology 80\/} (2018), 1259--1291.

\bibitem{leukocyte}
{\sc Tranquillo, R.~T., and Lauffenburger, D.~A.}
\newblock Stochastic model of leukocyte chemosensory movement.
\newblock {\em Journal of mathematical biology 25\/} (1987), 229--262.

\end{thebibliography}

\end{document}